\documentclass[12pt]{amsart}

\usepackage{amsthm,amsfonts,amsmath,amsthm, amssymb}
\usepackage{verbatim}
\usepackage{graphics,graphicx}
\numberwithin{equation}{section}

\newcommand{\bea}{\begin{eqnarray}}
\newcommand{\eea}{\end{eqnarray}}
\newcommand{\be}{\begin{eqnarray*}}
\newcommand{\ee}{\end{eqnarray*}}
\newtheorem{theorem}{Theorem}[section]
\newtheorem{lemma}{Lemma}[section]
\newtheorem{corollary}{Corollary}[section]

\theoremstyle{remark}
\newtheorem{remark}{Remark}[section]

\begin{document}

\title{Toeplitz Block Matrices in Compressed Sensing}
\author{Florian Sebert, Leslie Ying, and Yi Ming Zou}
\thanks{F. Sebert and Y. M. Zou are with the Department of Mathematical Sciences, University of Wisconsin, Milwaukee, WI 53201, USA email: fmsebert@uwm.edu, ymzou@uwm.edu}\thanks{L. Ying is with the Department of Electrical Engineering, University of Wisconsin, Milwaukee, WI 53201, USA email: leiying@uwm.edu}
\date{01/10/2008}

\maketitle

\begin{abstract}
Recent work in compressed sensing theory shows that $n\times N$
independent and identically distributed (IID) sensing matrices whose
entries are drawn independently from certain probability
distributions guarantee exact recovery of a sparse signal with high
probability even if $n\ll N$. Motivated by signal processing
applications, random filtering with Toeplitz sensing matrices whose
elements are drawn from the same distributions were considered and
shown to also be sufficient to recover a sparse signal from reduced
samples exactly with high probability. This paper considers Toeplitz
block matrices as sensing matrices. They naturally arise in
multichannel and multidimensional filtering applications and include
Toeplitz matrices as special cases. It is shown that the probability
of exact reconstruction is also high. Their performance is validated
using simulations.
\end{abstract}

\section{Introduction}
The central problem in compressed sensing (CS) is
 the recovery of a
vector $x\in \mathbb{R}^N$ from its linear measurements $y$ of the
form \bea\label{linmeas} y_i =<x,\varphi_i>,\;1\le i\le n, \eea
where $n$ is assumed to be much smaller than $N$. Of course, for
$n\ll N$, (\ref{linmeas}) posts an under-determined system of
equations which has non-unique solutions. Exact recovery of the
original vector $x$ needs further prior information. The work by
Cand{\'{e}}s, Donoho, Romberg, Tao, and others (see e.g.
\cite{Cand1},\cite{Donoho}, and the references therein) showed that
under the assumption that $x$ is sparse, one can actually recover
$x$ from a sample $y$ which is much smaller in size than $x$ by
solving a convex program with a suitably chosen sampling basis
$\varphi_i, 1\le i\le n$. If we write the linear system
(\ref{linmeas}) in the form \bea y=\Phi x,\quad \mbox{where $\Phi$
is an $n\times N$ matrix}, \eea then the question about what
sampling methods guarantee the exact recovery of $x$ becomes the
question about what matrices are ``good'' compressed sensing
matrices, meaning that they ensure exact recovery of a sparse $x$
from $y$ with high probability under the condition that $n\ll N$.
\par
In \cite{Cand3} Cand\`{e}s and Tao introduce the \textit{restricted
isometry property} as a condition on matrices $\Phi$ which provides
a guarantee on the performance of $\Phi$ in compressed sensing.
\par
Following their definition, we say that a matrix
$\Phi\in\mathbb{R}^{n\times N}$ satisfies RIP of order
$m\in\mathbb{N}$ and constant $\delta_m\in(0,1)$ if
\begin{equation}\label{RIP1}
(1-\delta_m)\|z\|_2^2\le\|\Phi_T z\|_2^2\le(1+\delta_m)\|z\|_2^2
\qquad \forall z\in\mathbb{R}^{|T|},
\end{equation}
where $T\subset\{1,2,\ldots,N\}$, $|T|\le m$, and $\Phi_T$ denotes
the matrix obtained by retaining only the columns of $\Phi$
corresponding to the entries of $T$.
\par
It was shown in \cite{Cand3} (reinterpreted in \cite{Cohen}) that if
$\Phi$ satisfies RIP of order $3m$ and constant $\delta_{3m}\in(0,
1)$:
\begin{equation}\label{RIP}
 (1-\delta_{3m})\|z\|_2^2\le\|\Phi_T z\|_2^2\le(1+\delta_{3m})\|z\|_2^2 \qquad \forall
 z\in\mathbb{R}^{|T|},
\end{equation}
where $T\subset\{1,2,\ldots,N\}$ and $|T|\le 3m$, the decoder given
by
\begin{equation}
\triangle(y):=\text{argmin} \|x\|_{l_1^N} \ \ \ \ \ \text{ subject
to }\  \Phi x=y
\end{equation}
ensures exact recovery of $x$ from $y$.
\par
Recently Baraniuk et al \cite{Baraniuk} showed that matrices whose
entries are drawn independently from certain probability
distribution $P$ satisfy RIP of order $m$ with probability
$\ge1-e^{-c_2 n}$ for every $\delta_{m}\in(0,1)$ provided that $n\ge
c_1 m \ln(N/m)$, where $c_1,c_2>0$ are some positive constants
depending only on $\delta_{m}$. Motivated by applications in signal
processing, Bajwa et al \cite{Nowak} considered (truncated)
Toeplitz-structured matrices whose entries are drawn from the same
probability distributions $P$ and showed that they satisfy RIP of
order $3m$ with probability $\ge1-e^{-c_2 n/m^2}$ for every
$\delta_{3m}\in(0,1)$ provided that $n\ge c_1 m^3\ln(N/m)$.
\par
Some examples of probability distributions that can be used in this
context have been studied in \cite{Achl}. They include
\begin{equation*}
r_{i,j}\sim N\left(0,\frac{1}{n}\right),
\end{equation*}
\begin{equation}\label{eqn:distr}
r_{i,j}=\left\{
\begin{array}{ccc}
\displaystyle\frac{1}{\sqrt{n}}&\text{with probability} &1/2\\
-\displaystyle\frac{1}{\sqrt{n}}&\text{with probability} &1/2\\
\end{array}
\right.,
\end{equation}
\begin{equation*}
r_{i,j}=\left\{
\begin{array}{ccc}
\displaystyle\sqrt{\frac{3}{n}}&\text{with probability} &1/6\\
0 &\text{with probability} &2/3\\
-\displaystyle\sqrt{\frac{3}{n}}&\text{with probability} &1/6
\end{array}
\right..
\end{equation*}\\
\par
Motivated by applications in multichannel sampling, in this paper we
will consider Toeplitz block matrices with elements in each block
drawn independently from one of the probability distributions in
(\ref{eqn:distr}) and some other block matrices with similar
structures. We show that such matrices also satisfy RIP of order
$3m$ for every $\delta_{3m}\in(0,1)$ with high probability, provided
that $n\ge c_1 lm\ln(N/m)$, where $l\le 3m(3m-1)$ and $c_1>0$ is
some positive constant depending only on $\delta_{3m}$. These
Toeplitz block matrices naturally represent the system equation
matrices in multichannel sampling applications where a single input
signal is recovered from output samples of multiple channels with
IID random filters. The result justifies the use of multichannel
over single-channel systems in compressed sensing. The advantages of
Toeplitz matrices pointed out in \cite{Nowak}, like e.g. efficient
implementations, also apply to the matrices considered in this
paper.

\section{Main Result}

\begin{theorem}\label{thm: toepl}
For Toeplitz block matrices of the form
\begin{equation}
\label{eqn:TB} \Phi=\left( \begin{array}{ccccc}
\Phi_k \ \ & \Phi_{k-1} \ \ & \ldots \ \ & \Phi_2\ \ & \Phi_1 \\
\Phi_{k+1} \ \ & \Phi_{k} \ \ & \ldots \ \ & \Phi_3\ \ & \Phi_2 \\
\vdots \ \ & \vdots \ \ & \ddots \ \ & \ddots\ \ & \vdots \\
\Phi_{k+l-1} \ \ & \Phi_{k+l-2} \ \ & \ldots \ \ &\ldots\ \ & \Phi_l
\end{array} \right)\in \mathbb{R}^{n\times N}
\end{equation}
with blocks $\Phi_i\in\mathbb{R}^{d\times e}$ whose elements are
drawn independently from one of the probability distributions in
(\ref{eqn:distr}), there exist constants $c_1,c_2>0$ depending only
on $\delta_{3m}\in(0,1)$, such that:
\begin{enumerate}
 \item[(i)] If $l\le 3m(3m-1)$, then for any $n\ge c_1 lm \ln(N/m)$, $\Phi$ satisfies RIP of order $3m$ for every $\delta_{3m}\in(0,1)$ with probability at least
\begin{equation*}
1-e^{-c_2n/l}.
\end{equation*}
 \item[(ii)] If $l> 3m(3m-1)$, then for any $n\ge c_1 m^3\ln(N/m)$, $\Phi$ satisfies RIP of order $3m$ for every $\delta_{3m}\in(0,1)$ with probability at least
\begin{equation*}
1-e^{-c_2n/m^2}.
\end{equation*}
\end{enumerate}
\end{theorem}

The above theorem gives the requirement for and probability of exact
reconstruction of a $3m$-sparse signal $x$ from a measurement $y$ if
Toeplitz block matrices are used. In particular it says, that if the
number of blocks ($l$) in one column of $\Phi$ does not exceed a
certain value depending only on the sparsity of the signal $x$, the
probability of perfect reconstruction is greater and the number of
required measurements is smaller than if $l$ is not bounded in this
way.
\par
As noted in \cite{Nowak, Tropp}, Toeplitz matrices naturally arise
in one-dimensional single-channel filtering applications where the
matrix elements are filter coefficients. Similarly, the Toeplitz
block matrices defined in (\ref{eqn:TB}) naturally arise in
one-dimensional multichannel sampling applications where the length
of the filter is at least $l$ points larger than that of the input
signal. The conventional multichannel sampling theorem states that
the sampling rate reduction over the single channel system cannot
exceed the number of channels for exact recovery. While Theorem
\ref{thm: toepl} suggests that multichannel systems with IID random
filters might be able to reduce the sampling rate by a factor higher
than the number of channels.
\par
We remark, that for other block matrices with similar structures,
the result in Theorem \ref{thm: toepl} also holds (see IV).

\section{Proof of Main Result}
Let $T\subset \{1,2,\ldots,N\}$. Denote by $\Phi_{T,i}$ the $i$-th
row of the matrix $\Phi_T$ obtained by retaining only those columns
of $\Phi$ corresponding to the elements in $T$, and let
$\Phi_{T,i}\cap\Phi_j$ denote the set of random variables common to
the $i$-th row of $\Phi_T$ and the $j$-th block of $\Phi$.

We note that, if (\ref{RIP}) holds for a set
$T\subset\{1,2,\ldots,N\}$, then it also holds for any
$\tilde{T}\subset T$. To prove that Toeplitz IID block matrices
satisfy RIP with high probability, it is therefore enough to
consider only those sets $T$ where $|T|=3m$.
\par
\begin{lemma}\label{lemma:max_dep} Define the sets $D_{T,i}$ by
$$D_{T,i}=\{j\in\{1,2,\ldots,n\}:
\Phi_{T,j}\text{ is stochastically dependent on } \Phi_{T,i}, j\neq
i \}.$$
\begin{itemize}
\item[(i)] If $T$ satisfies $|T|<\frac{1+\sqrt{1+4l}}{2}$, then
$|D_{T,i}|\le|T|(|T|-1)\le l-1.$
\item[(ii)] If $T$ satisfies $|T|\ge\frac{1+\sqrt{1+4l}}{2}$, then
$|D_{T,i}|\le l-1.$
\end{itemize}

\end{lemma}
\begin{proof}
Fix $\Phi_{T,i}$. $T$ defines a sequence $\displaystyle
\{r_{t_s}\}_{s=1}^{k}$, where $r_{t_s}$ is the number of columns
from block $\Phi_{t_s}$ in $T$. Thus $\sum_{s=1}^k r_{t_s}=|T|$.
Consider the number of rows that have dependency with the elements
in $\Phi_{t_s}\cap \Phi_{T,i}$. Since all elements inside a single
block are independent, there can be no dependencies within one
block. Moreover, because of the structure of the matrix $\Phi$,
there can be at most
$$\left\{\begin{array}{cc}
0 & \text{if } \Phi_{t_s}\cap \Phi_{T,i}=\emptyset\\
|T|-r_{t_s} & \text{if } \Phi_{t_s}\cap \Phi_{T,i}\neq\emptyset
\end{array}\right.$$
 rows outside the block $\Phi_{t_s}$
that depend on any element in $\Phi_{t_s}\cap \Phi_{T,i}$.\\
(i) If $T$ satisfies $|T|<\frac{1+\sqrt{1+4l}}{2}$, i.e. if
$l>|T|(|T|-1)$, these rows may be distinct, and we have
\begin{align*}
|D_{T,i}|&\le\sum_{\{t_s, s\in\{1,2,\ldots,k\}:\Phi_{t_s}\cap
\Phi_{T,i}\neq\emptyset\}} (|T|-r_{t_s})\\
&\le \sum_{t\in T} (|T|-1)=|T|(|T|-1)\le l-1
\end{align*}
dependent rows.\\
(ii) If $T$ satisfies $|T|\ge \frac{1+\sqrt{1+4l}}{2}$, i.e. if
$|T|(|T|-1)\ge l$, then $|D_{T,i}|$ is upper bounded by the number
of blocks, so $|D_{T,i}|\le l-1$.

\end{proof}
In \cite{Achl} it has been shown that for given $n$, $N$, and
$T\subset\{1,2,\ldots,N\}$ with $|T|\le m$, an IID matrix of size
$n\times N$ with entries drawn independently from one of the
distributions $P$ in (\ref{eqn:distr})\footnote{These matrices
consist of columns whose squared norm is equal to 1 in expectation.}
satisfies (\ref{RIP1}) with probability
\begin{equation}\label{eqn:iidcase}
\ge 1-e^{-f(n,m,\delta_{m})},
\end{equation}
where
\begin{equation}\label{eqn:f}
f(n,m,\delta_{m})=c_0n-m\ln(12/\delta_{m})-\ln(2).
\end{equation}
\par
Now consider a (truncated) Toeplitz block matrix
$\Phi\in\mathbb{R}^{n\times N}$ as in (\ref{eqn:TB}), where the
blocks $\{\Phi_i\}_{i=1}^{k+l-1}$ are such IID matrices
$\in\mathbb{R}^{d\times e}$ with entries drawn independently from
the same set of distributions as above.
\par
The following lemma gives an upper bound for the probability that a
matrix as in (\ref{eqn:TB}) with $1\le l\le n$ satisfies (\ref{RIP})
for any fixed subset $T$ with $|T|=3m$. Lemma \ref{lemma:T-l-small}
gives a tighter bound for the case $l>|T|(|T|-1)$.
\par
The proof of Lemma \ref{lemma:T} uses an argument similar to the one
in the proof of Lemma 1 in \cite{Nowak}.

\begin{lemma}\label{lemma:T}
For given $T\subset\{1,2,\ldots,N\}$ with $|T|=m$, and
$\delta_{m}\in(0,1)$, the Toeplitz block submatrix $\Phi_T$
satisfies (\ref{RIP}) with probability at least
\begin{equation*}
1-e^{-f(d,m,\delta_{m})+\ln(l)}.
\end{equation*}
\end{lemma}

\begin{proof}
We can write the matrix $\Phi_{T}$ as
\begin{equation}
\Phi_{T}=\left(
\begin{array}{c}
\Phi_T^1\\
\vdots\\
\Phi_T^l
\end{array}
\right),
\end{equation}
where the blocks $\Phi_T^i$ of size $d\times |T|$ are given by the
columns determined by $T$ in the $i$-th row of blocks
$(\Phi_{k+i-1}, \Phi_{k+i-2}, \ldots, \Phi_{i})$ in $\Phi$.
\par
Note that $\forall i\in\{1,2,\ldots,l\},\ \Phi_T^i$ is an IID matrix
with entries from one of the distributions in (\ref{eqn:distr}). If
we let $\tilde{\Phi}_T^i=\sqrt{l}\Phi_T^i$, then the matrices
$\tilde{\Phi}_T^i$ have columns whose squared norm is equal to 1 in
expectation and by (\ref{eqn:iidcase}) satisfy (\ref{RIP}), i.e.
\begin{eqnarray*}
(1-\delta_{m})\|z\|_2^2\le
\|\tilde{\Phi}_T^iz\|_2^2\le(1+\delta_{m})\|z\|_2^2, \\ \forall
z\in\mathbb{R}^{|T|},\ \forall i\in\{1,2,\ldots,l\},
\end{eqnarray*}
with probability at least
\begin{equation}\label{eqn:iid-est}
 1-e^{-f(d,m,\delta_{m})}.
\end{equation}
Now since
\begin{equation}
\|\Phi_Tz\|_2^2=\sum_{i=1}^{l}\|\Phi_T^iz\|_2^2=\sum_{i=1}^{l}\frac{1}{l}\|\tilde{\Phi}_T^iz\|_2^2
\end{equation}
and $\sum_{i=1}^{l}\frac{1}{l}=1$, we have
\begin{equation*}
(1-\delta_{m})\|z\|_2^2\le
\|\Phi_Tz\|_2^2\le(1+\delta_{m})\|z\|_2^2,\qquad \forall
z\in\mathbb{R}^{|T|}.
\end{equation*}
In other words, the event $E_1=\{\tilde{\Phi}_T^i \text{ satisfies
(\ref{RIP}) }\forall i \}$ implies the event\\ $E_2=\{\Phi_T \text{
satisfies }(\ref{RIP})\}$. Consequently,
\begin{align*}
P(E_2)&=1-P(E_2^c)\ge 1-P(E_1^c) \notag\\
&\ge 1-\sum_{i=1}^l P(\{\tilde{\Phi}_T^i
\text{ does not satisfy }(\ref{RIP})\})\notag\\
&\ge 1-\sum_{i=1}^l e^{-f(d,m,\delta_{m})}\qquad \text{(by
(\ref{eqn:iid-est})) }\\
&=1-e^{-f(d,m,\delta_{m})+\ln(l)}.\label{eqn:iid}
\end{align*}
\end{proof}

\begin{lemma}\label{lemma:T-l-small}
For given $\ T\subset\{1,2,\ldots,N\}$ with $|T|=m$, and
$\delta_{m}\in(0,1)$, if $l>|T|(|T|-1)$, the Toeplitz block
submatrix $\Phi_T$ satisfies (\ref{RIP}) with probability at least
\begin{equation*}
1-e^{-f(\lfloor n/q\rfloor,m,\delta_{m})+\ln(q)},
\end{equation*}
where $q=|T|(|T|-1)+1$.
\end{lemma}

\begin{proof}
Let $\Phi_{T,i}$ denote the $i$-th row of $\Phi_T$ and construct an
undirected dependency graph $G=(V,E)$ such that $V=\{1,2,\ldots,n\}$
and $$E=\{(i,i')\in V\times V:i\neq i', \Phi_{T,i}\text{ and
}\Phi_{T,i'}\text{ are dependent}\}.$$ By Lemma \ref{lemma:max_dep},
$\Phi_{T,i}$ can at most be dependent with $|T|(|T|-1)$ other rows.
Therefore, the maximum degree $\triangle$ of $G$ is given by
$\triangle\leq |T|(|T|-1)$, and using the Hajnal-Szemer\'{e}di
theorem on equitable coloring of graphs, we can partition $G$ using
$q=|T|(|T|-1)+1$ colors. Let $\{C_j\}_{j=1}^{q}$ be the different
color classes, then
\begin{equation*}
|C_j|=\lfloor n/q\rfloor \text{   or   } |C_j|=\lceil n/q \rceil.
\end{equation*}
 Now, let $\Phi_T^j$ be the $|C_j|\times|T|$ submatrix obtained
from $\Phi_T$ retaining the rows corresponding to the indices in
$C_j$ and define $\tilde{\Phi}_T^j=\sqrt{n/|C_j|}\Phi_T^j$. Then
\begin{equation}\label{eqn:norms-lsmall}
\forall z\in\mathbb{R}^{|T|},\qquad
\|\Phi_Tz\|_2^2=\sum_{j=1}^{q}\|\Phi_T^jz\|_2^2=\sum_{j=1}^{q}\frac{|C_j|}{n}\|\tilde{\Phi}_T^jz\|_2^2.
\end{equation}
Every $\tilde{\Phi}_T^j$ is a $|C_j|\times|T|$ IID matrix whose
columns have squared norm equal to 1 in expectation. By
(\ref{eqn:iidcase}), they satisfy (\ref{RIP}) with probability at
least
\begin{equation}\label{eqn:iid-est-lsmall}
1-e^{-f(|C_j|,m,\delta_{m})}\ge 1-e^{-f(\lfloor
n/q\rfloor,m,\delta_{m})}.
\end{equation}
Since $\sum_{j=1}^{q}\frac{|C_j|}{n}=1$, by
(\ref{eqn:norms-lsmall}), we have that if
\begin{equation*}
(1-\delta_{m})\|z\|_2^2\le
\|\tilde{\Phi}_T^jz\|_2^2\le(1+\delta_{m})\|z\|_2^2,\qquad \forall
z\in\mathbb{R}^{|T|},\ \forall j\in\{1,2,\ldots,q\}
\end{equation*}
then
\begin{equation*}
(1-\delta_{m})\|z\|_2^2\le
\|\Phi_Tz\|_2^2\le(1+\delta_{m})\|z\|_2^2,\qquad \forall
z\in\mathbb{R}^{|T|}.
\end{equation*}
In other words, the event $E_1=\{\tilde{\Phi}_T^j \text{ satisfies
(\ref{RIP}) for all }j \}$ implies the event $E_2=\{\Phi_T \text{
satisfies }(\ref{RIP})\}$. Consequently,
\begin{align*}
P(E_2)&=1-P(E_2^c)\ge1-P(E_1^c)\notag\\
&\ge 1-\sum_{j=1}^q P(\{\tilde{\Phi}_T^j
\text{ does not satisfy }(\ref{RIP})\})\notag\\
&\ge 1-\sum_{j=1}^q e^{-f(\lfloor
n/q\rfloor,m,\delta_{m})}\qquad\text{(by (\ref{eqn:iid-est-lsmall}))}\\
&=1-e^{-f(\lfloor
n/q\rfloor,m,\delta_{m})+\ln(q)}.\label{eqn:iid-lsmall}
\end{align*}
\end{proof}

\noindent {\bf Main result in Theorem \ref{thm: toepl}.}
\begin{proof}
(i) From (\ref{eqn:f}) and Lemma \ref{lemma:T} we have that $\Phi$
satisfies (\ref{RIP}) for any $T\subset\{1,2,\ldots,N\}$ such that
$|T|=3m$ with probability at least
\begin{equation}
1-e^{-c_0d+3m\ln(12/\delta_{3m})+\ln(2)+\ln(l)}.
\end{equation}
Since there are ${N\choose 3m}\le(eN/3m)^{3m}$ such subsets, using
Bonferroni's inequality (see e.g. \cite{Oxford}) yields that $\Phi$
satisfies RIP of order $3m$ with probability at least
\begin{equation}\label{eqn:exponential}
1-e^{-c_0n/l+3m[\ln(12/\delta_{3m})+\ln(N/3m)+1]+\ln(2)+\ln(l)}.
\end{equation}
Fix $c_2>0$ and pick $c_1=(3\ln((12/\delta_{3m}))+15)/(c_0-c_2)$.
Then for any $n\ge c_1 lm \ln(N/m)$, the exponent of $e$ in
(\ref{eqn:exponential}) is upper bounded by $-c_2n/l$:
\begin{align*}
&-\frac{c_0n}{l}+3m\left[\ln\left(\frac{12}{\delta_{3m}}\frac{N}{3m}\right)+1\right]+\ln(2l)\le
 -\frac{c_2n}{l}
 \notag \\
 \Leftrightarrow\ \ &3m\left[\ln\left(\frac{12}{\delta_{3m}}\frac{N}{3m}\right)+1\right]+\ln(2l)\le \frac{n}{l}(c_0-c_2)\notag\\
 \Leftrightarrow\ \ &\frac{3lm}{c_0-c_2}\left[\ln\left(\frac{12}{\delta_{3m}}\frac{N}{3m}\right)+1+\frac{\ln(2)}{3m}+\frac{\ln(l)}{3m}\right]\le
 n\notag\\
 \Leftrightarrow\ \ &\frac{3lm\ln\left(\frac{N}{m}\right)}{c_0-c_2}\left[\frac{\ln\left(\frac{12}{\delta_{3m}}\frac{N}{3m}\right)+1+\ln(2)+\ln(l)}{3m\ln\left(\frac{N}{m}\right)}\right]\le
 n\notag\\
  \Leftarrow\ \ &\frac{3lm\ln\left(\frac{N}{m}\right)}{c_0-c_2}\left[\ln\left(\frac{12}{\delta_{3m}}\right)+5\right]\le n\notag\\
 \Leftrightarrow\ \ &c_1 lm\ln\left(\frac{N}{m}\right)\le n\\
\end{align*}
\par
(ii) From (\ref{eqn:f}) and Lemma \ref{lemma:T-l-small} we have that
$\Phi$ satisfies (\ref{RIP}) for any $T\subset\{1,2,\ldots,N\}$ such
that $|T|=3m$ with probability at least
\begin{align}
&1-e^{-c_0\lfloor n/q \rfloor+3m\ln(12/\delta_{3m})+\ln(2)+\ln(q)}\notag\\
 \ge &1-e^{-c_0 n/9m^2 +3m\ln(12/\delta_{3m})+\ln(2)+\ln(9m^2)+c_0}.
\end{align}
Since there are ${N\choose 3m}\le(eN/3m)^{3m}$ such subsets, using
Bonferroni's inequality again yields that $\Phi$ satisfies RIP of
order $3m$ with probability at least
\begin{equation}\label{eqn:exponential-l-small}
1-e^{-c_0k/9m^2+3m[\ln(12/\delta_{3m})+\ln(N/3m)+1]+\ln(2)+\ln(9m^2)+c_0}.
\end{equation}
Now fix $c_2>0$ and pick $c_1>27c_3/(c_0-9c_2)$, where
$c_3=\ln(12/\delta_{3m})+\ln(2)+c_0+4$. Then, for any $n\ge c_1
m^3\ln(N/m)$, the exponent of $e$ in (\ref{eqn:exponential-l-small})
is upper bounded by $-c_2n/m^2$. This completes the proof of the
theorem.
\end{proof}

\begin{remark}
If $l=1$, then $\Phi$ is an IID matrix, and Theorem \ref{thm: toepl}
lower bounds the probability of $\Phi$ satisfying RIP of order $3m$
by $1-e^{-c_2n}$, which recovers the bound obtained in
\cite{Baraniuk}.
\end{remark}
\begin{remark} \label{remark1}
As long as $l\le 3m(3m-1)$, a matrix $\Phi$ as in (\ref{eqn:TB})
satisfies RIP of order $3m$ with probability $1-e^{-c_2n/l}\ge
1-e^{c_2'n/m^2}$, which is the bound given in \cite{Nowak}, since
\begin{equation}
-c_2n/l\le-c_2n/(9m^2-3m)\le-c_2n/9m^2=-c_2'n/m^2.
\end{equation}
\end{remark}

\section{Other Block Matrices}
\subsection{Circular matrices}
The above consideration can be applied to (truncated) circulant
block matrices of the form
\begin{equation}
\label{eqn:CB} \Phi=\left( \begin{array}{ccccc}
\Phi_k \ \ & \Phi_{k-1} \ \ & \ldots \ \ & \Phi_2\ \ & \Phi_1 \\
\Phi_{1} \ \ & \Phi_{k} \ \ & \ldots \ \ & \Phi_3\ \ & \Phi_2 \\
\vdots \ \ & \vdots \ \ & \ddots \ \ & \ddots\ \ & \vdots \\
\Phi_{l-1} \ \ & \Phi_{l-2} \ \ & \ldots \ \ &\ldots\ \ & \Phi_l
\end{array} \right)\in \mathbb{R}^{n\times N},
\end{equation}
where the blocks $\Phi_i$ are all IID matrices.
\par
Similar to (\ref{eqn:TB}), the circulant matrices in (\ref{eqn:CB})
also represent the system equation matrices in multichannel
sampling, but the convolution is a circular one.  They usually arise
in applications where convolutions are implemented by
multiplications in Fourier domain.
\par
Before we present the theorem for this type of matrices, we first
comment on the maximum number of stochastically dependent rows in a
(truncated) circulant matrix of the form \begin{align}\label{eqn:C}
A=\left(
\begin{array}{ccccc}
a_q \ \ & a_{q-1} \ \ & \ldots \ \ & a_2\ \ & a_1 \\
a_{1} \ \ & a_{q} \ \ & \ldots \ \ & a_3\ \ & a_2 \\
\vdots \ \ & \vdots \ \ & \ddots \ \ & \ddots\ \ & \vdots \\
a_{p-1} \ \ & a_{p-2} \ \ & \ldots \ \ &\ldots\ \ & a_p
\end{array} \right)\in \mathbb{R}^{p\times q}.
\end{align}
Again, we denote by $A_{T,i}$ the $i$-th row of the matrix $A_T$,
which is obtained by retaining only those columns of $A$
corresponding to $T\subset\{1,2,\ldots,N\}$.
\begin{lemma}\label{lemma:max_dep_circ} Define the sets $D_{T,i}$ by
$D_{T,i}=\{j\in\{1,2,\ldots,p\}: A_{T,j}$\text{ is stochastically
dependent on } $A_{T,i}, j\neq i \}.$ Then $D_{T,i}$ has cardinality
at most $|T|(|T|-1)$.
\end{lemma}
\begin{proof}
Note first, that an upper bound for the case $p=q$ clearly upper
bounds the case where $p<q$. We may therefore assume that $p=q$ and
$A$ is a square circulant matrix. Then the number of rows
stochastically dependent on $A_{T,i}$ is independent of $i$ and we
can, w.l.o.g., assume that $i=1$. Let $\textbf{t}\in\{0,1\}^q$ be a
$q$-tuple defined by
\begin{equation*}
\textbf{t}_j=\left\{
\begin{array}{cc}
0 & \text{if }j\not\in T\\
1 & \text{if }j\in T \end{array} \right.,\quad\mbox{j=1,\ldots,q},
\end{equation*}
and consider the matrix
\begin{equation}\label{A-tilde}
\tilde{A}= \left(
\begin{array}{c}
\mathbf{t} \\
\sigma(\mathbf{t})\\
 \hdots \\
  \sigma^{q-1}(\mathbf{t})
\end{array}
\right) \in \mathbb{R}^{q\times q},
\end{equation}
where $\sigma:\{0,1\}^q\rightarrow \{0,1\}^q$ defines the
right-shift
$(\mathbf{t}_1,\ldots,\mathbf{t}_{q-1},\mathbf{t}_q)\rightarrow(\mathbf{t}_q,\mathbf{t}_1,\ldots,\mathbf{t}_{q-1})$.
Denote by $\tilde{A}_T$ the matrix obtained by retaining only those
columns of $\tilde{A}$ corresponding to $T\subset\{1,2,\ldots,q\}$.
It is now easy to see that
\begin{align*}
|D_{T,i}|&=|\{\tilde{A}_{T,i},i\in\{2,\ldots,q\}\ :\
h(\tilde{A}_{T,1},\tilde{A}_{T,i})<|T|\}|\\
&\le \{\#\text{ of ones in }\textbf{t}\}\cdot(\{\#\text{ of ones in
}\textbf{t}\}-1)\\
&=|T|(|T|-1),
\end{align*}
where $h:\{0,1\}^{q}\times\{0,1\}^{q}\rightarrow\mathbb{N}$ is the
Hamming distance defined by
\begin{equation*}
h(x,y) = |\{j\in\{1,2,\ldots,q\}\ :\ x_j \neq y_j\}|.
\end{equation*}
\end{proof}

The following theorem gives lower bounds for the probability that a
circulant block matrix as in \ref{eqn:C} satisfies the RIP of order
$3m$. Note that the bounds obtained are the same as in \ref{thm:
toepl} although the number of independent entries in $\Phi$ is
greater than before. This is due to the nature of the proof using
the number of stochastically dependent rows of $\Phi$ which is the
same for both Toeplitz and circulant matrices.

\begin{theorem}\label{thm:circ}
Let $\Phi$ be as in (\ref{eqn:CB}). Then there exist constants
$c_1,c_2>0$ depending only on $\delta_{3m}\in(0,1)$, such that:
\begin{enumerate}
 \item[(i)] If $l\le 3m(3m-1)$, then for any $n\ge c_1 lm \ln(N/m)$, $\Phi$ satisfies RIP of order $3m$ for every $\delta_{3m}\in(0,1)$ with probability at least
\begin{equation*}
1-e^{-c_2n/l}.
\end{equation*}
 \item[(ii)] If $l>3m(3m-1)$, then for any $n\ge c_1 m^3\ln(N/m)$, $\Phi$ satisfies RIP of order $3m$ for every $\delta_{3m}\in(0,1)$ with probability at least
\begin{equation*}
1-e^{-c_2n/m^2}.
\end{equation*}
\end{enumerate}
\end{theorem}
\begin{proof}
A similar argument as the one in the proof of Lemma
\ref{lemma:max_dep} shows that the upper bound for the maximum
number of rows stochastically dependent on any row of a (truncated)
circulant block matrix is the same as for the (truncated) Toeplitz
block matrices (use Lemma \ref{lemma:max_dep_circ}). Then the proof
of Theorem \ref{thm: toepl} directly applies to the setting at hand.
\end{proof}

\subsection{Circulant-circulant Matrices}
We also consider matrices that are (truncated) circulant block
matrices whose blocks are themselves circulant:
\begin{align}
\label{eqn:CB2} \Phi&=\left( \begin{array}{ccccc}
\Phi_k \ \ & \Phi_{k-1} \ \ & \ldots \ \ & \Phi_2\ \ & \Phi_1 \\
\Phi_{1} \ \ & \Phi_{k} \ \ & \ldots \ \ & \Phi_3\ \ & \Phi_2 \\
\vdots \ \ & \vdots \ \ & \ddots \ \ & \ddots\ \ & \vdots \\
\Phi_{l-1} \ \ & \Phi_{l-2} \ \ & \ldots \ \ &\ldots\ \ & \Phi_l
\end{array} \right)\in \mathbb{R}^{n\times N},\\
\label{eqn:CBC} \Phi_i&=\left( \begin{array}{ccccc}
\varphi_p^i \ \ & \varphi_{p-1}^i \ \ & \ldots \ \ & \varphi_2^i\ \ & \varphi_1^i \\
\varphi_{1}^i \ \ & \varphi_{p}^i \ \ & \ldots \ \ & \varphi_3^i\ \ & \varphi_2^i \\
\vdots \ \ & \vdots \ \ & \ddots \ \ & \ddots\ \ & \vdots \\
\varphi_{q-1}^i \ \ & \varphi_{q-2}^i \ \ & \ldots \ \ &\ldots\ \ &
\varphi_q^i
\end{array} \right)\in \mathbb{R}^{q\times p}.
\end{align}
Denote by $\tau:\{0,1\}^{kp}\rightarrow \{0,1\}^{kp}$ the
right-shift of \textit{blocks} $\Phi_i$ and by
$\sigma:\{0,1\}^{kp}\rightarrow \{0,1\}^{kp}$ the right-shift of
\textit{elements inside a block} $\Phi_i$, both by one position.
These matrices arise in two-dimentional imaging applications where
the independent elements are the coefficients of the point spread
function of the imaging system. Replacing (\ref{A-tilde}) in the
proof of Lemma \ref{lemma:max_dep_circ} by
\begin{equation*}
\bar{A}=\left(
\begin{array}{c}
\textbf{t}\\
\sigma^1\tau^0(\textbf{t})\\
\vdots\\
\sigma^{(i-1)(\text{mod }p)}\tau^{\lfloor\frac{i-1}{p}\rfloor}(\textbf{t})\\
\vdots\\
\sigma^{p-1}\tau^{l-1}(\textbf{t})
\end{array}
\right)\in \mathbb{R}^{lq\times kp},
\end{equation*}
readily yields the upper bound $|T|(|T|-1)$ for the number of rows
stochastically dependent on any one row of $\Phi$. Applying Lemma
\ref{lemma:T-l-small} and Theorem \ref{thm:circ} shows that the
probability for perfect reconstruction is no less than
$1-e^{-c_2n/m^2}$. This says that imaging systems with IID random
point spread functions can significantly reduce the number of
acquired samples, while still being able to reconstruct the original
sparse image if the above conditions hold.

\subsection{Circulant-circulant Block Matrices}
As a generalization of the matrices defined by (\ref{eqn:CB2}) and
(\ref{eqn:CBC}), the following matrices are also considered:
\begin{align}
\label{eqn:TCB} \Phi&=\left( \begin{array}{ccccc}
\Phi_{k_1} \ \ & \Phi_{{k_1}-1} \ \ & \ldots \ \ & \Phi_2\ \ & \Phi_1 \\
\Phi_{1} \ \ & \Phi_{{k_1}} \ \ & \ldots \ \ & \Phi_3\ \ & \Phi_2 \\
\vdots \ \ & \vdots \ \ & \ddots \ \ & \ddots\ \ & \vdots \\
\Phi_{{l_1}-1} \ \ & \Phi_{{l_1}-2} \ \ & \ldots \ \ &\ldots\ \ &
\Phi_{l_1}
\end{array} \right)\in \mathbb{R}^{n\times N},\\
\Phi_i&=\left( \begin{array}{ccccc}
\Upsilon_{k_2} \ \ & \Upsilon_{{k_2}-1} \ \ & \ldots \ \ & \Upsilon_2\ \ & \Upsilon_1 \\
\Upsilon_{1} \ \ & \Upsilon_{k_2} \ \ & \ldots \ \ & \Upsilon_3\ \ & \Upsilon_2 \\
\vdots \ \ & \vdots \ \ & \ddots \ \ & \ddots\ \ & \vdots \\
\Upsilon_{{l_2}-1} \ \ & \Upsilon_{{l_2}-2} \ \ & \ldots \ \
&\ldots\ \ & \Upsilon_{l_2}
\end{array} \right),\notag
\end{align}
where the blocks $\Upsilon_j$ are all IID matrices. These matrices
arise in multichannel two-dimensional imaging applications where the
number of rows in $\Upsilon_j$ corresponds to the $n/(l_1l_2)$
independent channels. We show next that these matrices are also good
compressed sensing matrices.
\begin{corollary}
Let $\Phi$ be as in (\ref{eqn:TCB}). Then there exist constants
$c_1,c_2>0$ depending only on $\delta_{3m}\in(0,1)$, such that:
\begin{enumerate}
 \item[(i)] If $l_1 l_2\le 3m(3m-1)$, then for any $n\ge c_1 l_1 l_2m \ln(N/m)$, $\Phi$ satisfies RIP of order $3m$ for every $\delta_{3m}\in(0,1)$ with probability at least
\begin{equation*}
1-e^{-c_2n/l_1 l_2}.
\end{equation*}
 \item[(ii)] If $l_1 l_2>3m(3m-1)$, then for any $n\ge c_1 m^3\ln(N/m)$, $\Phi$ satisfies RIP of order $3m$ for every $\delta_{3m}\in(0,1)$ with probability at least
\begin{equation*}
1-e^{-c_2n/m^2}.
\end{equation*}
\end{enumerate}
\end{corollary}
This follows directly from Lemma \ref{lemma:max_dep_circ} and
Theorem \ref{thm:circ}.

\subsection{Deterministic Construction}
The CS matrices we have considered so far are based on randomized
constructions. However, in certain applications, deterministic
constructions are preferred. In \cite{DeVore} DeVore provided a
deterministic construction of CS matrices using polynomials over
finite fields. We will consider deterministic block matrices based
on DeVore's construction. Let us first recall the construction in
\cite{DeVore}.
\par
Consider the set $\mathbb{Z}_p\times \mathbb{Z}_p$, where
$\mathbb{Z}_p$ denotes the field of integers modulo $p$, $p$ a
prime. This set has $n:=p^2$ elements. Define
$P_r:=\{f\in\mathbb{Z}_p[x]:\deg(f)\le r\}$, $0<r<p$. This set has
$N:=p^{r+1}$ elements. For every $f\in P_r$, define the graph of $f$
by
$$\mathcal{G}(f)=\{(x,y)\in\mathbb{Z}_p \times
\mathbb{Z}_p:y=f(x), x\in\mathbb{Z}_p\}\subset\mathbb{Z}_p\times
\mathbb{Z}_p$$
 and consider the column vector $v(f)\in\{0,1\}^n$, indexed
by the elements of $\mathbb{Z}_p\times \mathbb{Z}_p$ ordered
lexicographically, given by
$$v(f):=(\textbf{1}_{(0,0)\in\mathcal{G}(f)},
\ldots,\textbf{1}_{(0,p-1)\in\mathcal{G}(f)},
\textbf{1}_{(1,0)\in\mathcal{G}(f)},
\ldots,\textbf{1}_{(p-1,p-1)\in\mathcal{G}(f)})^t,$$ where
$$\textbf{1}_{(a,b)\in\mathcal{G}(f)}=\left\{\begin{array}{cc}1 & \text{if }(a,b)\in\mathcal{G}(f)\\
0 & \text{if }(a,b)\not\in\mathcal{G}(f)\end{array}\right.$$

Construct the matrix $\Phi_0=(v(f_1),v(f_2),\ldots,v(f_N))$, where
the polynomials $f_i$ are ordered lexicographically with respect to
their coefficients. It was shown in \cite{DeVore}, that the matrix
$\Phi=\frac{1}{\sqrt{p}}\Phi_0$ satisfies RIP for any $m<p/r+1$ with
$\delta=(m-1)r/p\ (<1)$.
\par
Now consider
\begin{equation}
\label{eqn:DB} \Psi_0=\left( \begin{array}{ccccc}
\Psi_t \ \ & \Psi_{t-1} \ \ & \ldots \ \ & \Psi_2\ \ & \Psi_1 \\
\Psi_{t+1} \ \ & \Psi_{t} \ \ & \ldots \ \ & \Psi_3\ \ & \Psi_2 \\
\vdots \ \ & \vdots \ \ & \ddots \ \ & \ddots\ \ & \vdots \\
\Psi_{t+s-1} \ \ & \Psi_{t+s-2} \ \ & \ldots \ \ &\ldots\ \ & \Psi_s
\end{array} \right)\in \mathbb{R}^{sp^2\times tl},
\end{equation}
where $tl\le p^{r+1}$, and each block $\Psi\in\mathbb{R}^{p^2\times
l}$ is constructed from the first $tl$ vectors $v(f),\ f\in P_r$, as
above.
\begin{theorem}
The matrix $\Psi=\frac{1}{\sqrt{sp}}\Psi_0$ satisfies RIP with
$\delta=(m-1)r/p$ for any $m<p/r+1$.
\end{theorem}

\begin{proof}
As before, we only have to consider the case where $|T|=m$. Let
$T\subset\{1,2,\ldots,tl\}$ such that $|T|=m$, and let $\Psi_T$ be
the matrix obtained by retaining only those columns of $\Psi$
corresponding to the elements in $T$. Consider the matrix
$G_T=\Psi_T^t\Psi_T$. Since every column of $\Psi_0$ has exactly
$sp$ ones, the diagonal elements of $G_T$ are all one. An off
diagonal element of $G_T$ has the form
$g^T_{ij}=\sum_{x=1}^{s}\langle v_{x,i},v_{x,j}\rangle$, where
$i,j\in\{1,2,\ldots,m\}$, and $v_{x,i}$ denotes the vector
$(\Psi_{T,(x-1)n+1,i},\ldots,\Psi_{T,(x-1)n+n,i})^t\in\{0,1\}^{n}$
that represents some polynomial $f\in P_r$. Since the graphs of two
different polynomials in $P_r$ have at most $r$ elements in common,
$g^T_{ij}\le sr/sp=r/p$ for any $i\neq j$. Therefore, the sum of all
off diagonal elements in any row or column of $G_T$ is
$\le(m-1)r/p=\delta<1$ whenever $m<p/r+1$. We can, therefore, write
\begin{equation}
G_T=I+B_T,
\end{equation}
where $\|B_T\|_1\le\delta$ and $\|B_T\|_\infty\le\delta$. Since
$\|B_T\|_2^2\le\|B_T\|_1\|B_T\|_\infty$, we have that
$\|B_T\|_2\le\delta$ and so the spectral norms of $B_T$ and
$B_T^{-1}$ are $\le 1+\delta$ and $\le(1-\delta)^{-1}$,
respectively. This shows that $\Psi$ satisfies (\ref{RIP}).
\end{proof}

\section{Numerical Results}
To validate that the probability of exact recovery for Toeplitz
block CS matrices is high, the performance of Toeplitz block, IID,
and Toeplitz CS matrices is compared empirically. In our simulation,
a length n = 2048 signal with randomly placed m = 20 non-zero
entries drawn independently from the Gaussian distribution was
generated. Each such generated signal is sampled using $n\times N$
IID, Toeplitz and Toeplitz block matrices with entries drawn
independently from the Bernoulli distribution and reconstructed
using the log barrier solver from \cite{l1magic}. The experiment is
declared a success if the signal is exactly recovered, i.e., the
error is within the range of machine precision. The empirical
probability of success is determined by repeating the reconstruction
experiment 1000 times and calculating the fraction of success. This
empirical probability of success is plotted as a function of the
number of measurement samples n in Fig. 1. The simulation results
show, that in the vast majority of applications all Toeplitz block
matrices perform similar to IID matrices.

\begin{figure}
\begin{center}
\includegraphics[width=10cm, bb=0 0 561 420]{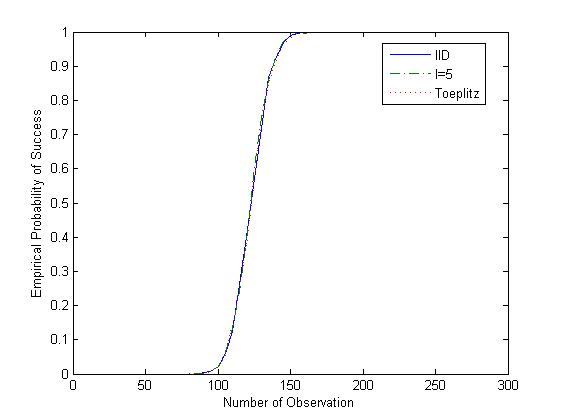}
\label{fig:comparison} \caption{Empirical probability of success
plotted against the number of observations for IID, Toeplitz block,
and Toeplitz matrices.}
\end{center}
\end{figure}

\section*{Acknowledgment}
The third author would like to acknowledge the support from IMA for
his participation in the short course ``Compressive Sampling and
Frontiers in Signal Processing".

\end{document}